\documentclass[preprint,review,1p]{elsarticle}

\usepackage{epsfig}
\usepackage[numbers]{natbib}
\usepackage{amssymb}
\usepackage{amsmath}
\usepackage{amsthm}
\usepackage{graphicx}%
\usepackage{multirow}%
\usepackage{amsmath,amssymb,amsfonts}%
\usepackage{amsthm}%
\usepackage{mathrsfs}%
\usepackage[title]{appendix}%
\usepackage{xcolor}%
\usepackage{textcomp}%
\usepackage{manyfoot}%
\usepackage{booktabs}%
\usepackage{algorithm}%
\usepackage{algorithmicx}%
\usepackage{algpseudocode}%
\usepackage{listings}%
\usepackage{array}
\usepackage{siunitx}
\usepackage{arydshln}
\usepackage{nicematrix}
\usepackage{comment}
\usepackage{enumerate}

\newtheorem{theorem}{Theorem}
\newtheorem{lem}{\bf Lemma} 
 
\newtheorem{corollary}{\bf Corollary}
\newtheorem{remark}{Remark}%

\definecolor{mypink}{rgb}{0.81, 0.1, 0.1}
\definecolor{myred}{rgb}{0.8, 0.1, 0.4}
\definecolor{myred1}{rgb}{0.8, 0.1, 0.2}
\definecolor{mygreen}{rgb}{0.1, 0.9, 0.05}
\definecolor{rbg}{RGB}{219, 48, 122}
\definecolor{mycyn}{cmyk}{0.9, 0.1, 0.4429, 0.1412}
\definecolor{mygray}{gray}{0.6}
\definecolor{plum}{rgb}{0.5, 0.3, 0.8}
\definecolor{purple}{rgb}{0.2, 0.1, 1}

\begin{document}

\begin{frontmatter}



\title{Nonlinear Observer Design in Discrete-time Systems: Incorporating LMI Relaxation Strategies} 

\author[label1]{Shivaraj Mohite\corref{cor1}}
\ead{shivaraj.mohite@univ-lorraine.fr}
\affiliation[label1]{organization={Mechanical and process engineering department, Rhineland-Palatinate Technical University of Kaiserslautern-Landau},
city={Kaiserslautern},
postcode={67663},
country={Germany}}


\begin{abstract}
This manuscript focuses on the $\mathcal{H}_\infty$ observer design for a class of nonlinear discrete systems under the presence of measurement noise or external disturbances. Two new Linear Matrix Inequality (LMI) conditions are developed in this method through the utilization of the reformulated Lipschitz property, a new variant of Young inequality and the well-known Linear Parameter Varying (LPV) approach. One of the key components of the proposed LMIs is the generalized matrix multipliers. The judicious use of these multipliers enables us to introduce more numbers of decision variables inside LMIs than the one illustrated in the literature. It aids in adding some extra degrees of freedom from a feasibility point of view, thus enhancing the LMI conditions. Thus, the established LMIs are less conservative than existing ones. Later on, the effectiveness of the developed LMIs and observer is highlighted through a numerical example and the application of state of charge (SoC) estimation in the Li-ion battery model.

\end{abstract}




\end{frontmatter}



\section{Introduction}\label{sec 1}

Over the past three decades, the topic of observer design for dynamical systems has received a significant amount of interest from researchers of control system engineering. This is because the state variables are essential in the system analysis and feedback control design. Although many of these variables can be measured through sensors, some remain inaccessible due to the unavailability of sensors. Thus, the observers play a pivotal role in modern-day applications that assist in capturing real-time information of systems, for example, the state-of-charge estimation of LI-ion battery model~\cite{shen_2023_SoC_model}, autonomous vehicle tracking~\cite{zemouche2017circle},  and so on. 

In the literature, numerous approaches have been proposed for linear observer design, and all have shown high reliability.
However, in contrast to the linear observer, the development of nonlinear observers is an arduous problem. There is no systematic technique to construct these types of observers. Hence, an abundant amount of research has been carried out in this area. The several popular methods of nonlinear observers are as follows: 
\begin{enumerate}[a)]
    \item Transformation-based observers~\cite{transformation-based};
    \item High-gain observers~\cite{ahrens2009high};
    \item Sliding-mode observers~\cite{spurgeon2008sliding};
    \item Linear Matrix Inequality (LMI)-based observers~\cite{2012_Nonlinear_Discrete_LMI}.
\end{enumerate}

All these methodologies are developed for continuous-time nonlinear systems, while the few methodologies related to discrete-time systems are showcased in~\cite{hu2012h,ibrir2007circle_descrete, LMI_discrete_2012} and~\cite{2012_Nonlinear_Discrete_LMI}. 
Among these techniques, LMI-based techniques are extensively studied by researchers and some of them are provided in~\cite{zemouche2006observer,ibrir2007circle_descrete, abbaszadeh2009lmi} and~\cite{LMI_discrete_2012}. Recently, the authors of~\cite{Shiv_IFAC_new} had established a new matrix-multiplier-based LMI approach for continuous-time nonlinear observer design. The LMI condition presented in~\cite{Shiv_IFAC_new} is less conservative than the one shown in~\cite{zemouche2017circle}. It inspires the authors to derive a new LMI condition for the development of nonlinear observers for discrete-time systems through the utilization of new matrix multipliers, reformulated Lipschitz property and Young inequalities. Further, the effectiveness of the developed LMI condition is shown through a numerical example. The performance of the observer is validated through the application of State-of-charge (SoC) estimation of Li-ion battery.

The rest of the paper is structured as follows: Some preliminaries related to the nonlinear observer design and the notations are illustrated in Section \ref{sec 2 not and pre}. Further, Section \ref{sec 3 prob statement} encompasses the contextualization of the problem statement. Segment \ref{Sec 4 LMI formulation} contains the formulation of two new LMI conditions. Further, some comments and remarks related to the proposed methodology are outlined in Section~\ref{sec 5 comment}. The effectiveness of the new LMI conditions and the observer performance is showcased in Section \ref{sec 6 num examples}. Finally, Section \ref{sec 7 conclusion} entails a few concluding remarks on the established approach.

\section{Nomenclature and Prerequisites}
\label{sec 2 not and pre}
This section encompasses the illustration of denotation used in this paper. Later on, we recaptured some mathematical tools
related to nonlinear observer design.
\subsection{Notations}\label{sec 2.1 notes}
Through this paper, the ensuing terminologies are utilized:
The euclidean norm and the $\mathcal{L}_2$ norm of a vector $e$ are depicted by $||e||$ and $||e||_{\mathcal{L}_2}$, respectively. The term $e_0$ denotes the initial values of $e(t)$ at $t=0$.
A vector of the canonical basis of $\mathbb{R}^{s}$ is illustrated as: $$e_s(i) = (\underbrace{0,\hdots,0,\overbrace{1}^{i^{\text{th}}}, 0, \hdots, 0}_{s \,\, \text{components}})^\top \in \mathbb{R}^{s},~s\geq 1.$$
The identity matrix and the null matrix are represented by $\mathbb{I}$ and $\mathbb{O}$, respectively.
The transpose of matrix $A$ is symbolised as $A^\top$, while, $A \in \mathbf{S}^{n}$ infers that $A$ is a symmetric matrix of dimension $n \times n$. The repeated blocks within a symmetric matrix are showcased by the symbol $(\star)$. 
$\lambda_{\min}(A) $ and $\lambda_{\max}(A) $ depict minimum and maximum eigenvalues of $A \in \mathbf{S}^{n}$, respectively.
For the aforementioned matrix $A$, $A > 0$ ($A < 0$) indicates that it is a positive definite matrix (a negative definite matrix). Similarly, a positive semi-definite matrix (a negative semi-definite matrix) is showcased by $A \geq 0$ ($A \leq 0$).  A block-diagonal matrix having elements $A_1, \hdots, A_n$ in the diagonal is described as $A = \text{block-diag}(A_1, \hdots,A_n)$.   

\subsection{Preliminaries}\label{sec 2.2 pre}
This segment presents an overview of the mathematical tools and background results which will be needed in the development of the LMI conditions.

\begin{lem}[Reformulated Lipschitz property\cite{zemouche2013lmi}]\label{Lem 1}
Let us consider a global Lipschitz nonlinear function $h: \mathbb{R}^n \to \mathbb{R}^n$. Then, for all $i,j \in \{1,\hdots,n\}$, there exists functions $h_{ij} : \mathbb{R}^n \times \mathbb{R}^n \to \mathbb{R}$ such that $\forall\,\Psi,\,\Phi \in \mathbb{R}^n,$
\begin{equation}
\begin{split}
 h(\Psi)-h(\Phi)&= \sum^{n}_{i=1}
       \sum^{n}_{j=1}h_{ij} \mathcal{H}_{ij} (\Psi-\Phi)\label{L 2.2}, 
\end{split}
    \end{equation}
where $\mathcal{H}_{ij} = e_n(i)e^\top_n(j)$ and $h_{ij}  \triangleq h_{ij}(\Psi^{\Phi_{j-1}},\Psi^{\Phi_{j}})$. Additionally, the functions $h_{ij}$ hold
 \begin{equation}
       h_{{{ij}_{\min}}}\leq
      h_{ij}  \leq
    h_{{{ij}_{\max}}},\label{L 2.3}
    \end{equation}
where $h_{{{ij}_{\min}}}$ and $h_{{{ij}_{\max}}}$ are constants.
\end{lem}
\begin{lem}[Young's inequalities]\label{Lem 2}
For any two vectors $X, Y \in \mathbb{R}^n$ and a matrix $Z >0 \in \mathbf{S}^{n}$, the ensuing inequalities are true:  
\begin{equation}\label{L 3.1}
    X^\top Y+Y^\top X \leq X^\top Z^{-1}X + Y^\top ZY,
\end{equation}
and
\begin{equation}\label{L 3.2}
  X^\top Y+Y^\top X \leq (X+ZY)^\top (2Z)^{-1}(X+ZY). 
\end{equation}
The inequality~\eqref{L 3.1} is known as the standard Young's inequality. However, the authors of~\cite{zemouche2017circle} presented a variant of Young's inequality, represented by~\eqref{L 3.2}.
\end{lem}
\begin{lem}[~\cite{Shiv_IFAC_new}]\label{Lem 3}
 Let us consider 
\begin{align}
\mathbb{X}^\top &= \begin{bmatrix}
a_1\mathbb{I}_n & a_2\mathbb{I}_n & \hdots & a_n\mathbb{I}_n
\end{bmatrix}, \label{Lem3 1}\\
\mathbb{Y}^\top &= \begin{bmatrix}
b_1\mathbb{I}_n & b_2\mathbb{I}_n & \hdots &b_n\mathbb{I}_n
\end{bmatrix},\label{Lem3 2}
\end{align}
along with
\begin{equation}\label{Lem3 3}
    Z = \begin{bmatrix}
Z_1 & Z_{a^2_1} & \hdots& Z_{a^n_1}\\
\star &  Z_2 & \hdots &  Z_{a^n_2}\\
\star & \star & \ddots & \vdots\\
\star & \star & \hdots  & Z_n\\
\end{bmatrix},
\end{equation}
where $0 \leq a_i \leq b_i~\forall i \in \{ 1,\hdots,n \}$ and $Z_i> 0 \in\mathbf{S}^{{n}},~Z_{a^j_i} \geq 0 \in\mathbf{S}^{{n}}$ $\forall \,\,i \in \{ 1,\hdots,n \}$ so that $Z>0$. Then, the subsequent inequality is fulfilled:
\begin{equation}\label{Lem3 4}
    \mathbb{X}^\top Z\mathbb{X} - \mathbb{Y}^\top Z \mathbb{Y} \leq 0. 
\end{equation}
\end{lem}
For the proof of Lemma~\ref{Lem 3}, one can refer~\cite{Shiv_IFAC_new}.   

\begin{lem}\label{Lem 4}
For any given matrices $A, B$ and $C$ of appropriate dimension, the subsequent equality is true:
\begin{equation}\label{L 3.3}
    \begin{bmatrix}
        A^\top C A&A^\top C B\\
B^\top C A& B^\top P B
    \end{bmatrix}=\begin{bmatrix}
       A&B
    \end{bmatrix}^\top C\begin{bmatrix}
       A&B
    \end{bmatrix}
\end{equation}
\end{lem}
The proof of this lemma is very straightforward. Through matrix multiplication, one can easily derive~\eqref{L 3.3}.

\section{Problem statement}\label{sec 3 prob statement}
Let us consider that the following sets of equations represent a class of disturbance-affected nonlinear system dynamics with nonlinear output:
\begin{align}\label{eq 1}
\begin{split}
     x_{k+1}  &= A x_k + G f(x_k)+B_1 u_k+E \omega_k, \\
     y_k       &= C x_k+ F g(x_k)+B_2 u_k+D\omega_k, 
\end{split}
\end{align}
where $x \in \mathbb{R}^n$ and $y \in \mathbb{R}^p$ denote a state vector of the system, and its output, respectively. The input provided to the system is illustrated by $u\in\mathbb{R}^s$. However, $\omega\in \mathbb{R}^q$ depicts $\mathcal{L}_2$ bounded noise/disturbance vectors affecting the system dynamics and measurements. The matrices $A \in \mathbb{R}^{n \times n}$, $G \in \mathbb{R}^{n \times m}$, $B_1\in \mathbb{R}^{n \times s},\,B_2\in \mathbb{R}^{p \times s}$, $C \in R^{p \times n} $, $E\in \mathbb{R}^{n \times q}$, $D\in \mathbb{R}^{p \times q}$ 
 and $F\in \mathbb{R}^{p \times r}$ are known and constant. The functions $f(\cdot): \mathbb{R}^n \to \mathbb{R}^m$ and $g(\cdot): \mathbb{R}^n \to \mathbb{R}^r$ are the nonlinearities present in the dynamics and outputs of the system, respectively. Both functions are presumed to be globally Lipschitz. Further, we have expressed $f(\cdot)$ and $g(\cdot)$ in the ensuing manner:
 \begin{subequations}
  \begin{equation}\label{eq 2 f}
f(x_k) = \begin{bmatrix}
 f_1( F_1 x_k) \\\vdots
 \\f_i (\underbrace{F_i x_k}_{\theta_i})\\
\vdots \\
f_m (\theta_m) \end{bmatrix}, 
\end{equation}
  \begin{equation}\label{eq 2 g}
g(x_k) = \begin{bmatrix}
 g_1( G_1 x_k) \\
\vdots \\f_i (\underbrace{G_i x_k}_{\nu_i})\\
\vdots \\
g_r (\nu_r) 
\end{bmatrix},
\end{equation}
\end{subequations}
where $F_i \in \mathbb{R}^{\bar{n} \times n}~\forall~i\in\{1,\hdots ,m\}$ and $G_i \in \mathbb{R}^{\bar{p} \times n}~\forall~i\in\{1,\hdots ,r\}$.
\begin{remark}\label{Rem 1 noise}
If system dynamics and output of~\eqref{eq 1} are influenced by two different noises, $\omega_1$ and $\omega_2$, through $E_1$ and $D_1$ respectively, then the system can be rewritten in the form of~\eqref{eq 1} by considering matrices $E = \begin{bmatrix}E_1 & \mathbb{O}\end{bmatrix},~D = \begin{bmatrix}\mathbb{O} & D_1\end{bmatrix},$ and the noise vector $\omega = \begin{bmatrix}\omega_1\\\omega_2\end{bmatrix}$.    
\end{remark}

For the state estimation purposes of the system~\eqref{eq 1}, we have deployed the subsequent Luenberger observer form:
\begin{equation}\label{eq 3}
\begin{split}
   \hat{x}_{k+1}=A\hat{x}_k+G f(\hat{x}_k)+B_1 u_k+L\big(y_k-C\hat{x}_k-B_2 u_k -F g(\hat{x}_k)\big),
\end{split}
\end{equation}
where $\hat{x}_k$ and $L \in  \mathbb{R}^{n \times p}$ are the estimated states and the observer gain matrix, respectively. The estimation error of the proposed observer~\eqref{eq 3} is defined as $$e_k={x}_{k}-\hat{x}_{k}.$$ 
Thus, from~\eqref{eq 1} and~\eqref{eq 3}, the estimation error dynamic is computed and expressed as:
\begin{equation}\label{eq 4}
\begin{split}
  e_{k+1}&=(A-LC)e_k+G \big(f(x_k)-f(\hat{x}_k)\big)
  -LF\big(g(x_k)-g(\hat{x}_ k)\big)+(E-LD)\omega_k.  
\end{split}
 \end{equation}

Since $f(\cdot)$ and $g(\cdot)$ are globally Lipschitz, the implementation of Lemma~\ref{Lem 1} on the terms $(f(x_k)-f(\hat{x}_k))$ and $(g(x_k)-g(\hat{x}_k))$ yields:
\begin{enumerate}[1)]
    \item There exist functions $f_{ij} : \mathbb{R}^{\bar{n}} \times \mathbb{R}^{\bar{n}} \to \mathbb{R}$, $g_{ij} : \mathbb{R}^{\bar{p}} \times \mathbb{R}^{\bar{p}} \to \mathbb{R}$ such that
\begin{subequations}\label{eq 6}
     \begin{equation}\label{eq 6 f_tilde}
 f(x_k)-f(\hat{x}_k)=\sum_{i,j=1}^{m,\bar{n}}   f_{ij}\mathcal{H}_{ij}F_i e_k, \end{equation}
\begin{equation}\label{eq 6 g_tilde}
  g(x_k)-g(\hat{x}_k)=\sum_{i,j=1}^{r,\bar{p}}   g_{ij}\mathcal{G}_{ij}G_i e_k, 
\end{equation}
 \end{subequations}
where $f_{ij} \triangleq f_{ij}(\theta_i^{\hat{\theta}_{i,j-1}},\theta_i^{\hat{\theta}_{i,j}})$ and $g_{ij} \triangleq g_{ij}(\nu_i^{\hat{\nu}_{i,j-1}},\nu_i^{\hat{\nu}_{i,j}})$. 
\item The functions $f_{ij}$ and $g_{ij}$ satisfy
\begin{align}
    {f}_{a_{ij}} &\leq  f_{ij}\leq {f}_{b_{ij}},\label{eq 7.1 f_tilde}\\
    {g}_{a_{ij}} &\leq  g_{ij}\leq {g}_{b_{ij}},\label{eq 7.1 g_tilde}
\end{align}
where ${f}_{a_{ij}}$, ${f}_{b_{ij}}$, ${g}_{a_{ij}}$, and ${g}_{b_{ij}}$ are known constants.
\end{enumerate}
Without loss of generality, let us presume that ${f}_{a_{ij}}=0$ and ${g}_{a_{ij}}=0$. Thus, the inequalities~\eqref{eq 7.1 f_tilde} and~\eqref{eq 7.1 g_tilde} are reformulated as:
\begin{align}
    0 &\leq  f_{ij}\leq {f}_{b_{ij}},\label{eq 7 f_tilde}\\
    0 &\leq  g_{ij}\leq {g}_{b_{ij}}.\label{eq 7 g_tilde}
\end{align}
One can refer~\cite{zemouche2013lmi} for additional information about this.

By incorporating~\eqref{eq 6} into~\eqref{eq 4}, the error dynamic~\eqref{eq 4} is reformulated as:
\begin{equation}\label{eq 8}
\begin{split}
  e_{k+1}&=\underbrace{\bigg((A-LC)+\sum_{i,j=1}^{m,\bar{n}}   f_{ij}G\mathcal{H}_{ij}H_i-\sum_{i,j=1}^{r,\bar{p}}   g_{ij}LF\mathcal{G}_{ij}G_i\bigg)}_{\mathbb{A}} e_k+\underbrace{\big(E-LD\big)}_{\mathbb{E}} \omega_k.
\end{split}
 \end{equation}
\begin{remark}\label{rem 1}
In various practical applications, it is possible to have  $f_{a_{ij}},\,g_{a_{ij}}  \neq 0$. In such cases, the system~\eqref{eq 8} is rewritten as
\begin{equation}\nonumber
\begin{split}
e_{k+1}&=\bigg(\underbrace{
A-LC+\sum_{i,j=1}^{m,\bar{n}}f_{a_{ij}}G\mathcal{H}_{ij}H_i-\sum_{i,j=1}^{r,\bar{p}} g_{a_{ij}}L F\mathcal{G}_{ij}G_i}_{\Tilde{{A}}}
\bigg)e_k \\&+  \bigg(\sum_{i,j=1}^{m,\bar{n}} \underbrace{(f_{ij}-f_{a_{ij}})}_{\tilde{f}_{ij}} G\mathcal{H}_{ij}H_i \Tilde{x}-\sum_{i,j=1}^{r,\bar{p}} \underbrace{(g_{ij}-g_{a_{ij}})}_{\tilde{g}_{ij}} L F\mathcal{G}_{ij}G_i \bigg)e_k+\mathbb{E}\omega_k.
\end{split}
\end{equation}
It yields:
\begin{equation}\label{eq 8.1}
\begin{split}
e_{k+1}= \bigg(\underbrace{\Tilde{{A}}+\sum_{i,j=1}^{m,\bar{n}} {\tilde{f}_{ij}} G\mathcal{H}_{ij}H_i \Tilde{x}- \sum_{i,j=1}^{r,\bar{p}} {\tilde{g}_{ij}} L F\mathcal{G}_{ij}G_i}_{\Tilde{\mathbb{A}}} \bigg)e_k+\mathbb{E}\omega_k.
\end{split}
\end{equation}
For the error dynamic~\eqref{eq 8.1}, the functions
$\tilde{f}_{ij}$ and $ \tilde{g}_{ij}$ hold~\eqref{eq 7 f_tilde} and~\eqref{eq 7 g_tilde}, respectively. It is easy to notice that both forms~\eqref{eq 8} and~\eqref{eq 8.1} are analogous. 
\end{remark}
The objective of the proposed methodology is to estimate the gain matrix $L$ so that
\begin{enumerate}
    \item The estimation error dynamic~\eqref{eq 8} is asymptotically stable in the absence of the disturbances/noise, i.e., at $\omega=0$.
    \item When $\omega\neq0$, the closed-loop system~\eqref{eq 8} fulfills the ensuing  $\mathcal{H}_\infty$ criterion:
\begin{equation}\label{eq H criterion 1}
        ||e||_{\mathcal{L}_2}  \leq \sqrt{\mu ||\omega||^2_{\mathcal{L}_2}+\nu ||e_0||^2},
\end{equation}
where the positive scalar $\sqrt{\mu}$ is known as the noise attenuation level.
\end{enumerate}

The aforementioned problem statement has garnered a significant amount of interest from researchers in the domain of control system engineering, resulting in the establishment of numerous LMI-based methods, for example,~\cite{ibrir2007circle_descrete},~\cite{Khadilja_2019_rnc},~\cite{chu2018_SMO}, and so on. The LMI conditions provided by each of these methods are based on several mathematical tools, such as the Schur lemma and the Young inequality. Though all of these techniques yield less conservative LMI conditions, there is a scope for further enhancements. In the sequel, through the exploration of Lemma~\ref{Lem 2}, Lemma~\ref{Lem 3} and newly defined matrix multipliers, two new LMI criteria are derived.
\section{Main result}\label{Sec 4 LMI formulation}
This section of the manuscript is devoted to the formulation of the $\mathcal{H}_\infty$ criterion-based LMI conditions which ensures the asymptotic stability of the error dynamic~\eqref{eq 8}.

For the stability analysis of the system~\eqref{eq 8}, the following quadratic Lyapunov function is used:
\begin{equation}\label{Ly fun}
 V(e_k) = e_k^\top P e_k,~\text{where}~P> 0\in \mathbf{S}^{n}   
\end{equation}
Let us consider $$\Delta V_k=V(e_{k+1})-V(e_k).$$ 
Through the utilization of~\eqref{eq 8} and~\eqref{Ly fun}, one can obtain:
\begin{equation}\label{sec2 del VK}
    \begin{split}
\Delta V_k&=e_k^\top\Big(-P+\mathbb{A}^\top P \mathbb{A}
\Big)e_k+e_k^\top\Big(\mathbb{A}^\top P \mathbb{E}\Big)\omega_k
+\omega_k^\top \Big(\mathbb{E}^\top P\mathbb{A}\Big)e_k
\\&+\omega_k^\top \mathbb{E}^\top P \mathbb{E} \omega_k.
    \end{split}
\end{equation}
According to~\cite{Khadilja_2019_rnc}, the  $\mathcal{H}_\infty$ criterion~\eqref{eq H criterion 1} is satisfied if the ensuing inequality is true:
\begin{equation}\label{eq H criterion 2}
    \mathcal{W}_k \triangleq \Delta V_k+ ||e_k||^2-\mu ||\omega_k||^2\leq 0.
\end{equation}
From~\eqref{sec2 del VK}, the inequality~\eqref{eq H criterion 2} is modified as:
\begin{equation}\label{eq W_cal 1}
    \begin{split}
\mathcal{W}_k&=e_k^\top\Big(\mathbb{I}-P+\mathbb{A}^\top P \mathbb{A}
\Big)e_k+e_k^\top \Big(\mathbb{A}^\top P \mathbb{E}\Big)\omega_k+\omega_k^\top \Big(\mathbb{E}^\top P\mathbb{A}\Big)e_k \\&+\omega_k^\top \Big(\mathbb{E}^\top P \mathbb{E} -\mu \mathbb{I}\Big)\omega_k.
    \end{split}
\end{equation}
Further, $\mathcal{W}_k \leq 0$ if
\begin{equation}\label{eq W_cal 2}
\begin{bmatrix}
 \mathbb{I}-P+ \mathbb{A}^\top P \mathbb{A}&\mathbb{A}^\top P\mathbb{E}\\  \mathbb{E}^\top P\mathbb{A}&\mathbb{E}^\top P\mathbb{E}-\mu \mathbb{I}
\end{bmatrix} \leq 0.
\end{equation} 
Through the deployment of Lemma~\ref{Lem 4}, the inequality~\eqref{eq W_cal 2} is equivalent to
\begin{equation}\label{eq W_cal 3} 
  \begin{bmatrix}
 \mathbb{I}-P&\mathbb{O}\\ \mathbb{O}&-\mu \mathbb{I}
\end{bmatrix} +  \begin{bmatrix}
        \mathbb{A}^\top\\\mathbb{E}^\top
    \end{bmatrix}P\begin{bmatrix}
       \mathbb{A}&\mathbb{E}
    \end{bmatrix} \leq 0.
\end{equation}
The use of Schur Lemma on~\eqref{eq W_cal 3} resulted in
\begin{equation}\label{eq W_cal 4}
 \Sigma_1  +{\mathbf{NL}}\leq 0, 
\end{equation}
where 
\begin{equation}\label{eq W_cal 4 Sigma1}
 \Sigma_1=   \begin{bmatrix}
   \begin{bmatrix}
 \mathbb{I}-P&\mathbb{O}\\ \mathbb{O}&-\mu \mathbb{I}
\end{bmatrix}&  \begin{bmatrix}
        (A-LC)^\top P\\ (E-LD)^\top P
    \end{bmatrix}\\
    \star &-P
  \end{bmatrix},
\end{equation} 
\begin{equation}\label{eq W_cal 4 NL}
   \mathbf{NL}= 
{{\begin{bmatrix}
\mathbb{O}&\mathbb{O}&\bigg(\sum_{i,j=1}^{m,\bar{n}}   f_{ij}G\mathcal{H}_{ij}H_i-\sum_{i,j=1}^{r,\bar{p}}   g_{ij}LF\mathcal{G}_{ij}G_i\bigg)^\top P\\
\star &\mathbb{O}&\mathbb{O}\\
\star &\star &\mathbb{O}
  \end{bmatrix}}}.
\end{equation}
For enhancement of comprehensibility of the method, let us introduce $R^\top = P L$. Now, one can express the term $\Sigma_1$ in the ensuing form:
\begin{equation}\label{eq LMI 1 Sigma}
 \Sigma=  \begin{bmatrix}
\begin{bmatrix}
 \mathbb{I}-P&\mathbb{O}\\ \mathbb{O}&-\mu \mathbb{I}
\end{bmatrix}& \begin{bmatrix}
        A^\top P- C^\top R\\ E^\top P- D^\top R
\end{bmatrix}\\
    \star &-P
  \end{bmatrix}.
\end{equation}
Similarly, the term $\mathbf{NL}$ is rewritten as:
\begin{equation}\label{eq W_cal NL}
\begin{split}
\mathbf{NL} &=\sum_{i,j=1}^{m,\bar{n}}\begin{pmatrix}
    \underbrace{\begin{bmatrix}
     \mathbb{O}\\\mathbb{O}\\ PG\mathcal{H}_{ij}  
\end{bmatrix}}_{\mathbb{U}_{ij}^\top}\underbrace{ f_{ij}\overbrace{\begin{bmatrix}
       H_i & \mathbb{O}& \mathbb{O}
\end{bmatrix}}^{\mathbb{H}_i}}_{\mathbb{V}_{ij}}+\mathbb{V}_{ij}^\top \mathbb{U}_{ij}
\end{pmatrix}\\&+\sum_{i,j=1}^{r,\bar{p}}\begin{pmatrix}\underbrace{\begin{bmatrix}
     \mathbb{O}\\\mathbb{O}\\ - R^\top F\mathcal{G}_{ij}  
\end{bmatrix}}_{\mathbb{M}_{ij}^\top}\underbrace{g_{ij}\overbrace{\begin{bmatrix}
        G_i & \mathbb{O}& \mathbb{O}
\end{bmatrix}}^{\mathbb{G}_i}}_{\mathbb{N}_{ij}}+\mathbb{N}_{ij}^\top \mathbb{M}_{ij}\end{pmatrix}.
\end{split}
\end{equation}

Recently, the topic of LMI-based observers has been extensively investigated to handle Lipschitz nonlinearities. One can go through~\cite{ibrir2007circle_descrete,2012_Nonlinear_Discrete_LMI,Khadilja_2019_rnc} and so on. The authors of \cite{2012_Nonlinear_Discrete_LMI, aviles2019observer} have used the global form of nonlinearities (i.e., $\tilde{f}(x,\hat{x})=f(x)-f(\hat{x})$). However, the detailed form of nonlinearities (that is,~\eqref{eq 3}) was deployed in papers\cite{Khadilja_2019_rnc}. The use of nonlinearities in their detailed form enables the inclusion of additional decision variables in the LMI approach.
In this paper, we propose two new LMI conditions inspired by the method outlined in~\cite{LCSS_Shiv_2023}.

\begin{table*}[!ht]
\hrule
{\small
\begin{equation}\label{eq 11 matrix Z}
    \mathbb{Z} =   \left[
    \begin{array}{c c c c c c c c c c c c c}
    Z_{11} & Z_{a^1_{12}} & \hdots&Z_{a^1_{1\bar{n}}} & \textcolor{plum}{Z_{b^{11}_{21}}} & \textcolor{plum}{Z_{{b}^{11}_{22}}}  &\textcolor{plum}{\hdots} &\textcolor{plum}{Z_{{b}^{11}_{2\bar{n}}}}   &\textcolor{green}{\hdots}& \textcolor{myred}{Z_{b^{11}_{m1}}}  & \textcolor{myred}{Z_{{b}^{11}_{m2}}} &\textcolor{myred}{\hdots} & \textcolor{myred}{Z_{{b}^{11}_{m\bar{n}}}}
    \\
    Z_{a^1_{12}} & Z_{12} & \hdots &Z_{a^2_{1\bar{n}}}&\textcolor{plum}{Z_{{b}^{12}_{21}}}& \textcolor{plum}{Z_{{b}^{12}_{22}}} &\textcolor{plum}{\hdots} & \textcolor{plum}{Z_{{b}^{12}_{2\bar{n}}}} &\textcolor{green}{\hdots} &
    \textcolor{myred}{Z_{{b}^{12}_{m1}}}& \textcolor{myred}{Z_{{b}^{12}_{m2}}}&\textcolor{myred}{\hdots} & \textcolor{myred}{Z_{{b}^{12}_{m\bar{n}}}}\\
\vdots&\vdots&\ddots&\vdots&\textcolor{plum}{\vdots}&\textcolor{plum}{\vdots}&\textcolor{plum}{\ddots}&\textcolor{plum}{\vdots}&\textcolor{green}{\hdots} &\textcolor{myred}{\vdots}&\textcolor{myred}{\vdots}&\textcolor{myred}{\ddots}&\textcolor{myred}{\vdots}\\
Z_{a^1_{1\bar{n}}} & Z_{a^2_{1\bar{n}}} & \hdots & Z_{1\bar{n}}&\textcolor{plum}{Z_{{b}^{1\bar{n}}_{21}}}& \textcolor{plum}{Z_{{b}^{1\bar{n}}_{22}}}
&\textcolor{plum}{\hdots} & \textcolor{plum}{Z_{{b}^{1\bar{n}}_{2\bar{n}}}}&\textcolor{green}{\hdots} &
\textcolor{myred}{Z_{{b}^{1\bar{n}}_{m1}}}& \textcolor{myred}{Z_{{b}^{1\bar{n}}_{m2}}}
&\textcolor{myred}{\hdots} & \textcolor{myred}{Z_{{b}^{1\bar{n}}_{m\bar{n}}}}
\\
\textcolor{plum}{Z_{b^{11}_{21}}} & \textcolor{plum}{Z_{b^{12}_{21}}}  &\textcolor{plum}{\hdots} &\textcolor{plum}{Z_{{b}^{1\bar{n}}_{21}}}& Z_{21} &   Z_{a^1_{22} }& \hdots& Z_{a^1_{2\bar{n}}}&\textcolor{green}{\hdots}& \textcolor{myred1}{Z_{{b}^{21}_{m1}}} & \textcolor{myred1}{Z_{{b}^{21}_{m2}}}
&\textcolor{myred1}{\hdots} & \textcolor{myred1}{Z_{{b}^{21}_{m\bar{n}}}}\\
\textcolor{plum}{Z_{b^{11}_{22}}} & \textcolor{plum}{Z_{b^{12}_{22}}}  &\textcolor{plum}{\hdots} &\textcolor{plum}{Z_{{b}^{1\bar{n}}_{22}}}& Z_{a^1_{22}}&  Z_{22} &\hdots &Z_{a^2_{2\bar{n}}} &\textcolor{green}{\hdots} &
\textcolor{myred1}{Z_{{b}^{22}_{m1}}} & \textcolor{myred1}{Z_{{b}^{22}_{m2}}}
&\textcolor{myred1}{\hdots} & \textcolor{myred1}{Z_{{b}^{22}_{m\bar{n}}}}  \\
\textcolor{plum}{\vdots}&
\textcolor{plum}{\vdots}&
\textcolor{plum}{\ddots}&
\textcolor{plum}{\vdots}&\vdots&\vdots&\ddots&\vdots&\textcolor{green}{\hdots}&\textcolor{myred1}{\vdots}&
\textcolor{myred1}{\vdots}&
\textcolor{myred1}{\ddots}&
\textcolor{myred1}{\vdots}\\
\textcolor{plum}{Z_{b^{11}_{2\bar{n}}}} & \textcolor{plum}{Z_{b^{12}_{2\bar{n}}}}  &\textcolor{plum}{\hdots} &\textcolor{plum}{Z_{{b}^{1\bar{n}}_{2\bar{n}}}}& 
 Z_{a^1_{2\bar{n}}}&
 Z_{a^2_{2\bar{n}}}&\hdots&  Z_{2\bar{n}}&\textcolor{green}{\hdots}&\textcolor{myred1}{Z_{{b}^{2\bar{n}}_{m1}}} & \textcolor{myred1}{Z_{{b}^{2\bar{n}}_{m2}}}
&\textcolor{myred1}{\hdots} & \textcolor{myred1}{Z_{{b}^{2\bar{n}}_{m\bar{n}}}}
\\
\textcolor{green}{\vdots}&
\textcolor{green}{\vdots}&
\textcolor{green}{\vdots}&
\textcolor{green}{\vdots}&
\textcolor{green}{\vdots}&
\textcolor{green}{\vdots}&
\textcolor{green}{\vdots}&
\textcolor{green}{\vdots}&\ddots&
\textcolor{green}{\vdots}&
\textcolor{green}{\vdots}&
\textcolor{green}{\vdots}&
\textcolor{green}{\vdots}\\
\textcolor{myred}{Z_{b^{11}_{m1}}} & \textcolor{myred}{Z_{b^{12}_{m1}}}  &\textcolor{myred}{\hdots} &\textcolor{myred}{Z_{{b}^{1\bar{n}}_{m1}}}&\textcolor{myred1}{Z_{b^{21}_{m1}}} & \textcolor{myred1}{Z_{b^{22}_{m1}}}  &\textcolor{myred1}{\hdots} &\textcolor{myred1}{Z_{{b}^{2\bar{n}}_{m1}}}&\textcolor{green}{\hdots} & Z_{m1} &  Z_{a^1_{m2}} &\hdots &   Z_{a^1_{m\bar{n}}}\\
\textcolor{myred}{Z_{b^{11}_{m2}}} & \textcolor{myred}{Z_{b^{12}_{m2}}}  &\textcolor{myred}{\hdots} &\textcolor{myred}{Z_{{b}^{1\bar{n}}_{m2}}}&\textcolor{myred1}{Z_{b^{21}_{m2}}} & \textcolor{myred1}{Z_{b^{22}_{m2}}}  &\textcolor{myred1}{\hdots} &\textcolor{myred1}{Z_{{b}^{2\bar{n}}_{m2}}}&\textcolor{green}{\hdots} &
 Z_{a^1_{m2}}&  Z_{m2} &\hdots &  Z_{a^2_{m\bar{n}}}\\
\textcolor{myred}{\vdots}&
\textcolor{myred}{\vdots}&
\textcolor{myred}{\ddots}&
\textcolor{myred}{\vdots}&\textcolor{myred1}{\vdots}&
\textcolor{myred1}{\vdots}&
\textcolor{myred1}{\ddots}&
\textcolor{myred1}{\vdots}&\textcolor{green}{\hdots}&\vdots&\vdots&\ddots&\vdots\\
\textcolor{myred}{Z_{b^{11}_{m\bar{n}}}} & \textcolor{myred}{Z_{b^{12}_{m\bar{n}}}}  &\textcolor{myred}{\hdots} &\textcolor{myred}{Z_{{b}^{1\bar{n}}_{m\bar{n}}}}&\textcolor{myred1}{Z_{b^{21}_{m\bar{n}}}} & \textcolor{myred1}{Z_{b^{22}_{m\bar{n}}}}  &\textcolor{myred1}{\hdots} &\textcolor{myred1}{Z_{{b}^{2\bar{n}}_{m\bar{n}}}}&\textcolor{green}{\hdots}&
 Z_{a^1_{m\bar{n}}}& Z_{a^2_{m\bar{n}}}
&\hdots &  Z_{m\bar{n}}
      \end{array}
\right],
\end{equation}}
where 
$Z_{ij}>0\in \mathbf{S}^{\bar{n}},~Z_{a^k_{ij}}\geq 0\in \mathbf{S}^{\bar{n}}\,\,\forall i,k\in \{1,\hdots,m\}, \& j\in \{1,\hdots,\bar{n}\}$;$~Z_{b^{kj}_{ij}}\geq0\in \mathbf{S}^{\bar{n}},\forall i\in \{2,\hdots,m\},k\in \{1,\hdots,m-1\}, \& j\in \{1,\hdots,\bar{n}\}$ such that $\mathbb{Z}>0$.
\hrule
\end{table*}
\begin{table*}[!ht]
\begin{equation}\label{eq 11 matrix S}
    \mathbb{S} =   \left[
    \begin{array}{c c c c c c c c c c c c c}
    S_{11} & S_{a^1_{12}} & \hdots&S_{a^1_{1\bar{p}}} & \textcolor{plum}{S_{b^{11}_{21}}} & \textcolor{plum}{S_{{b}^{11}_{22}}}  &\textcolor{plum}{\hdots} &\textcolor{plum}{S_{{b}^{11}_{2\bar{p}}}}   &\textcolor{green}{\hdots}& \textcolor{myred}{S_{b^{11}_{r1}}}  & \textcolor{myred}{S_{{b}^{11}_{r2}}} &\textcolor{myred}{\hdots} & \textcolor{myred}{S_{{b}^{11}_{r\bar{p}}}}
    \\
    S_{a^1_{12}} & S_{12} & \hdots &S_{a^2_{1\bar{p}}}&\textcolor{plum}{S_{{b}^{12}_{21}}}& \textcolor{plum}{S_{{b}^{12}_{22}}} &\textcolor{plum}{\hdots} & \textcolor{plum}{S_{{b}^{12}_{2\bar{p}}}} &\textcolor{green}{\hdots} &
    \textcolor{myred}{S_{{b}^{12}_{r1}}}& \textcolor{myred}{S_{{b}^{12}_{r2}}}&\textcolor{myred}{\hdots} & \textcolor{myred}{S_{{b}^{12}_{r\bar{p}}}}\\
\vdots&\vdots&\ddots&\vdots&\textcolor{plum}{\vdots}&\textcolor{plum}{\vdots}&\textcolor{plum}{\ddots}&\textcolor{plum}{\vdots}&\textcolor{green}{\hdots} &\textcolor{myred}{\vdots}&\textcolor{myred}{\vdots}&\textcolor{myred}{\ddots}&\textcolor{myred}{\vdots}\\
S_{a^1_{1\bar{p}}} & S_{a^2_{1\bar{p}}} & \hdots & S_{1\bar{p}}&\textcolor{plum}{S_{{b}^{1\bar{p}}_{21}}}& \textcolor{plum}{S_{{b}^{1\bar{p}}_{22}}}
&\textcolor{plum}{\hdots} & \textcolor{plum}{S_{{b}^{1\bar{p}}_{2\bar{p}}}}&\textcolor{green}{\hdots} &
\textcolor{myred}{S_{{b}^{1\bar{p}}_{r1}}}& \textcolor{myred}{S_{{b}^{1\bar{p}}_{r2}}}
&\textcolor{myred}{\hdots} & \textcolor{myred}{S_{{b}^{1\bar{p}}_{r\bar{p}}}}
\\
\textcolor{plum}{S_{b^{11}_{21}}} & \textcolor{plum}{S_{b^{12}_{21}}}  &\textcolor{plum}{\hdots} &\textcolor{plum}{S_{{b}^{1\bar{p}}_{21}}}& S_{21} &   S_{a^1_{22} }& \hdots& S_{a^1_{2\bar{p}}}&\textcolor{green}{\hdots}& \textcolor{myred1}{S_{{b}^{21}_{r1}}} & \textcolor{myred1}{S_{{b}^{21}_{r2}}}
&\textcolor{myred1}{\hdots} & \textcolor{myred1}{S_{{b}^{21}_{r\bar{p}}}}\\
\textcolor{plum}{S_{b^{11}_{22}}} & \textcolor{plum}{S_{b^{12}_{22}}}  &\textcolor{plum}{\hdots} &\textcolor{plum}{S_{{b}^{1\bar{p}}_{22}}}& S_{a^1_{22}}&  S_{22} &\hdots &S_{a^2_{2\bar{p}}} &\textcolor{green}{\hdots} &
\textcolor{myred1}{S_{{b}^{22}_{r1}}} & \textcolor{myred1}{S_{{b}^{22}_{r2}}}
&\textcolor{myred1}{\hdots} & \textcolor{myred1}{S_{{b}^{22}_{r\bar{p}}}}  \\
\textcolor{plum}{\vdots}&
\textcolor{plum}{\vdots}&
\textcolor{plum}{\ddots}&
\textcolor{plum}{\vdots}&\vdots&\vdots&\ddots&\vdots&\textcolor{green}{\hdots}&\textcolor{myred1}{\vdots}&
\textcolor{myred1}{\vdots}&
\textcolor{myred1}{\ddots}&
\textcolor{myred1}{\vdots}\\
\textcolor{plum}{S_{b^{11}_{2\bar{p}}}} & \textcolor{plum}{S_{b^{12}_{2\bar{p}}}}  &\textcolor{plum}{\hdots} &\textcolor{plum}{S_{{b}^{1\bar{p}}_{2\bar{p}}}}& 
 S_{a^1_{2\bar{p}}}&
 S_{a^2_{2\bar{p}}}&\hdots&  S_{2\bar{p}}&\textcolor{green}{\hdots}&\textcolor{myred1}{S_{{b}^{2\bar{p}}_{r1}}} & \textcolor{myred1}{S_{{b}^{2\bar{p}}_{r2}}}
&\textcolor{myred1}{\hdots} & \textcolor{myred1}{S_{{b}^{2\bar{p}}_{r\bar{p}}}}
\\
\textcolor{green}{\vdots}&
\textcolor{green}{\vdots}&
\textcolor{green}{\vdots}&
\textcolor{green}{\vdots}&
\textcolor{green}{\vdots}&
\textcolor{green}{\vdots}&
\textcolor{green}{\vdots}&
\textcolor{green}{\vdots}&\ddots&
\textcolor{green}{\vdots}&
\textcolor{green}{\vdots}&
\textcolor{green}{\vdots}&
\textcolor{green}{\vdots}\\
\textcolor{myred}{S_{b^{11}_{r1}}} & \textcolor{myred}{S_{b^{12}_{r1}}}  &\textcolor{myred}{\hdots} &\textcolor{myred}{S_{{b}^{1\bar{p}}_{r1}}}&\textcolor{myred1}{S_{b^{21}_{r1}}} & \textcolor{myred1}{S_{b^{22}_{r1}}}  &\textcolor{myred1}{\hdots} &\textcolor{myred1}{S_{{b}^{2\bar{p}}_{r1}}}&\textcolor{green}{\hdots} & S_{r1} &  S_{a^1_{r2}} &\hdots &   S_{a^1_{r\bar{p}}}\\
\textcolor{myred}{S_{b^{11}_{r2}}} & \textcolor{myred}{S_{b^{12}_{r2}}}  &\textcolor{myred}{\hdots} &\textcolor{myred}{S_{{b}^{1\bar{p}}_{r2}}}&\textcolor{myred1}{S_{b^{21}_{r2}}} & \textcolor{myred1}{S_{b^{22}_{r2}}}  &\textcolor{myred1}{\hdots} &\textcolor{myred1}{S_{{b}^{2\bar{p}}_{r2}}}&\textcolor{green}{\hdots} &
 S_{a^1_{r2}}&  S_{r2} &\hdots &  S_{a^2_{r\bar{p}}}\\
\textcolor{myred}{\vdots}&
\textcolor{myred}{\vdots}&
\textcolor{myred}{\ddots}&
\textcolor{myred}{\vdots}&\textcolor{myred1}{\vdots}&
\textcolor{myred1}{\vdots}&
\textcolor{myred1}{\ddots}&
\textcolor{myred1}{\vdots}&\textcolor{green}{\hdots}&\vdots&\vdots&\ddots&\vdots\\
\textcolor{myred}{S_{b^{11}_{r\bar{p}}}} & \textcolor{myred}{S_{b^{12}_{r\bar{p}}}}  &\textcolor{myred}{\hdots} &\textcolor{myred}{S_{{b}^{1\bar{p}}_{r\bar{p}}}}&\textcolor{myred1}{S_{b^{21}_{r\bar{p}}}} & \textcolor{myred1}{S_{b^{22}_{r\bar{p}}}}  &\textcolor{myred1}{\hdots} &\textcolor{myred1}{S_{{b}^{2\bar{p}}_{r\bar{p}}}}&\textcolor{green}{\hdots}&
 S_{a^1_{r\bar{p}}}& S_{a^2_{r\bar{p}}}
&\hdots &  S_{r\bar{p}}
      \end{array}
\right],
\end{equation}
where 
$S_{ij}>0\in \mathbf{S}^{\bar{p}},~S_{a^k_{ij}}\geq 0\in \mathbf{S}^{\bar{p}}\,\,\forall i,k\in \{1,\hdots,r\}, \& j\in \{1,\hdots,\bar{p}\}$;$~S_{b^{kj}_{ij}}\geq 0\in \mathbf{S}^{\bar{p}},\,\,\forall i\in \{2,\hdots,r\},k\in \{1,\hdots,r-1\}, \& j\in \{1,\hdots,\bar{p}\}$
so that $\mathbb{S}>0$.
\hrule
\end{table*}

In order to avoid cumbersome equations, the term $\mathbf{NL}$ is reformulated as:
\begin{equation}\label{eq W_cal NL 1}
    \mathbf{NL}=\mathbb{U}^\top (\mathbb{H}\Phi )+\Phi^\top\mathbb{H}^\top  \mathbb{U}+\mathbb{M}^\top ( \mathbb{G}\Psi)+\Psi^\top \mathbb{G}^\top \mathbb{M},
\end{equation}
where
\begin{align}
  \mathbb{U}&=\begin{bmatrix}
     \mathbb{U}^\top_{11} & \hdots&
     \mathbb{U}^\top_{1\bar{n}}
     &\hdots &
      \mathbb{U}^\top_{m1}
     &\hdots & 
     \mathbb{U}^\top_{m\bar{n}}
    \end{bmatrix}^\top, \label{eq 13 U}\\
    \mathbb{H}&=\text{block-diag}(\mathbb{H}_1,\hdots,\mathbb{H}_1,\hdots,\mathbb{H}_m,\hdots,\mathbb{H}_m),\label{eq 13 H}\\
    \Phi&=\begin{bmatrix}
     f_{11} \mathbb{I}&\hdots&
    f_{1\bar{n}} \mathbb{I}& \hdots&
    f_{m1} \mathbb{I}& \hdots & 
      f_{m\bar{n}}\mathbb{I}
    \end{bmatrix}^\top,
    \label{eq 13 Phi}\\
\mathbb{M}&=\begin{bmatrix} \mathbb{M}^\top_{11} & \hdots&\mathbb{M}^\top_{1\bar{p}}&\hdots & \mathbb{M}^\top_{r1}&\hdots & \mathbb{M}^\top_{r\bar{p}}\end{bmatrix}^\top,\label{eq 13 M} \\
 \mathbb{G}&=\text{block-diag}(\mathbb{G}_1,\hdots,\mathbb{G}_1,\hdots,\mathbb{G}_r,\hdots,\mathbb{G}_r),\label{eq 13 G}\\
    \Psi&=\begin{bmatrix}
     g_{11} \mathbb{I}&\hdots&
    g_{1\bar{p}} \mathbb{I}& \hdots&
    g_{r1} \mathbb{I}& \hdots & 
      g_{r\bar{p}}\mathbb{I}
    \end{bmatrix}^\top
    \label{eq 13 Psi}.
\end{align}
From~\eqref{eq W_cal 4},~\eqref{eq LMI 1 Sigma} and~\eqref{eq W_cal NL 1}, $\mathcal{W}\leq 0$ if
\begin{equation}\label{eq W_cal 5}
 \Sigma  +\underbrace{\mathbb{U}^\top (\mathbb{H}\Phi )+\Phi^\top\mathbb{H}^\top  \mathbb{U}}_{\mathbf{NL}_1}+\underbrace{\mathbb{M}^\top (\mathbb{G}\Psi )+\Psi^\top \mathbb{G}^\top \mathbb{M}}_{\mathbf{NL}_2}\leq 0.
\end{equation}
Through the use of Lemma~\ref{Lem 2} and Lemma~\ref{Lem 3}, two new LMI conditions are developed in the sequels, which ensures the asymptotic stability of the system~\eqref{eq 8}.

\begin{theorem}\label{Theorem 1 IFAC LMI}
Let us introduce two matrices, $\mathbb{Z}$ and $\mathbb{S}$, illustrated by~\eqref{eq 11 matrix Z} and~\eqref{eq 11 matrix S}, respectively. If there exist matrices $P>0\in \mathbf{S}^{n}$, $R\in \mathbb{R}^{p \times n}$ and a positive scalar $\mu$ such that the following optimization problem is solvable:
\begin{equation}\label{eq LMI 1}
    \begin{split}
\text{minimize}~\mu~\text{subject~to}&\\
        \begin{bmatrix}
           \Sigma & \mathbb{U}^\top &
    (\mathbb{Z}\mathbb{H}\Phi_m)^\top &\mathbb{M}^\top &
    (\mathbb{S}\mathbb{G}\Psi_m)^\top\\
           \star&-\mathbb{Z}&\mathbb{O}&
           \mathbb{O}&\mathbb{O}\\
\star&\star&-\mathbb{Z}&\mathbb{O}& \mathbb{O}\\
\star&\star&\star&-\mathbb{S}&\mathbb{O}\\
\star&\star&\star&\star&-\mathbb{S}
        \end{bmatrix}\leq 0,
    \end{split}
\end{equation}
where $\Sigma$, $\mathbb{U}$, and $\mathbb{M}$ are described by~\eqref{eq LMI 1 Sigma},~\eqref{eq 13 U} and~\eqref{eq 13 M}, respectively. 
Additionally,
\begin{align}
\Phi_m&=\begin{bmatrix}
     f_{b_{11}} \mathbb{I}&\hdots&
    f_{b_{1\bar{n}}} \mathbb{I}& \hdots&
    f_{b_{m1}} \mathbb{I}& \hdots & 
      f_{b_{m\bar{n}}}\mathbb{I}
    \end{bmatrix}^\top,
    \label{eq 13 Phim}\\ 
\Psi_m&=\begin{bmatrix}
     g_{b_{11}} \mathbb{I}&\hdots&
    g_{b_{1\bar{p}}} \mathbb{I}& \hdots&
    g_{b_{r1}} \mathbb{I}& \hdots & 
      g_{b_{r\bar{p}}}\mathbb{I}
    \end{bmatrix}^\top
    \label{eq 13 Psim}.
\end{align}
Then, the error dynamic~\eqref{eq 8} is $\mathcal{H}_\infty$ asymptotically stable. The gain matrix $L$ is computed by utilising $L=P^{-1}R^\top$.
\end{theorem}
\begin{proof}
The deployment of the inequality~\eqref{L 3.1} on the terms $\mathbf{NL}_1$ and $\mathbf{NL}_2$ yield:
\begin{equation}\label{thm 1 pf e 1}
    \mathbf{NL}_1 \leq \mathbb{U}^\top (\mathbb{Z})^{-1}\mathbb{U}+\Phi^\top\mathbb{H}^\top (\mathbb{Z})\mathbb{H}\Phi,
\end{equation}
and
\begin{equation}\label{thm 1 pf e 2}
    \mathbf{NL}_2 \leq \mathbb{M}^\top (\mathbb{S})^{-1}\mathbb{M}+\Psi^\top\mathbb{G}^\top (\mathbb{S})\mathbb{G}\Psi,
\end{equation}
where $\mathbb{Z}>0$ and $\mathbb{S}>0$ are defined in~\eqref{eq 11 matrix Z} and~\eqref{eq 11 matrix S}, respectively.

Since $f_{ij}\leq f_{b_{ij}}$, we obtain the following inequality by implementing Lemma~\ref{Lem 3} on the matrices $\Phi$ and $\Phi_m$:
\begin{equation}\label{thm 1 pf e 3}
    \Phi^\top\mathbb{H}^\top (\mathbb{Z})\mathbb{H}\Phi \leq 
    \Phi_m^\top\mathbb{H}^\top (\mathbb{Z})\mathbb{H}\Phi_m,
\end{equation}
where $\Phi_m$ is defined in~\eqref{eq 13 Phim}.
\\Similarly,
\begin{equation}\label{thm 1 pf e 4}
    \Psi^\top\mathbb{G}^\top (\mathbb{S})\mathbb{G}\Psi \leq 
    \Psi_m^\top\mathbb{G}^\top (\mathbb{S})\mathbb{G}\Psi_m,
\end{equation}
where $\Psi_m$ is described in~\eqref{eq 13 Psim}.
\\From~\eqref{thm 1 pf e 3} and~\eqref{thm 1 pf e 4}, we get
\begin{align}
 \mathbf{NL}_1 &\leq \mathbb{U}^\top (\mathbb{Z})^{-1}\mathbb{U}+\Phi_m^\top\mathbb{H}^\top (\mathbb{Z})\mathbb{H}\Phi_m,\\
\mathbf{NL}_2 &\leq \mathbb{M}^\top (\mathbb{S})^{-1}\mathbb{M}+\Psi_m^\top\mathbb{G}^\top (\mathbb{S})\mathbb{G}\Psi_m.  
\end{align}
Thus, the inequality~\eqref{eq W_cal 5} is satisfied if
\begin{equation}\label{eq W_cal 5.1}
 \Sigma  +\mathbb{U}^\top (\mathbb{Z})^{-1}\mathbb{U}+\Phi_m^\top\mathbb{H}^\top (\mathbb{Z})\mathbb{H}\Phi_m+\mathbb{M}^\top (\mathbb{S})^{-1}\mathbb{M}+\Psi_m^\top\mathbb{G}^\top (\mathbb{S})\mathbb{G}\Psi_m\leq 0.
\end{equation}
The LMI~\eqref{eq LMI 1} is deduced by deploying Schur's Lemma on~\eqref{eq W_cal 5.1}. If the LMI~\eqref{eq LMI 1} is feasible, then the condition specified in~\eqref{eq W_cal 3} is fulfilled. Thus, the estimation error dynamic~\eqref{eq 8} satisfies $\mathcal{H}_\infty$ criterion~\eqref{eq H criterion 1}, ensuring the asymptotic stability of~\eqref{eq 8}.
\end{proof}
In the next theorem, an LPV-based LMI condition is showcased.
\begin{theorem}\label{Theorem 2 LPV LMI}
The system~\eqref{eq 8} is $\mathcal{H}_\infty$ asymptotically stable if there exist matrices $\mathbb{Z}$ and $\mathbb{S}$ under the form~\eqref{eq 11 matrix Z} and~\eqref{eq 11 matrix S}, respectively, along with $P>0\in \mathbf{S}^{n}$, $R\in \mathbb{R}^{p \times n}$ and a positive scalar $\mu$, such that the following optimization problem is solvable:
\begin{equation}\label{eq LMI 2}
\begin{split}
\text{minimize}~\mu~\text{subject~to}&\\
    \begin{bmatrix}
    \Sigma & (\mathbb{U}+\mathbb{Z}\mathbb{H}\Phi)^\top&(\mathbb{M}+\mathbb{S}\mathbb{G}\Psi)^\top\\
        \star &-2\mathbb{Z}&\mathbb{O}\\
        \star& \star &-2\mathbb{S}
\end{bmatrix} &< 0,\,\,{ \forall \Phi\in\mathcal{F}_{{H}_m},\forall \Psi\in\mathcal{G}_{{H}_m}},
\end{split}
\end{equation}
where 
\begin{align}
\mathcal{F}_{{H}_m}&=\bigg\{
\{\mathcal{F}_{11},
\hdots,\mathcal{F}_{1\bar{n}},\hdots,\mathcal{F}_{m1},\hdots,\mathcal{F}_{m\bar{n}}\} : \mathcal{F}_{ij} \in [0,f_{b_{ij}}]
  \bigg\} ,\label{thm 2 eq FHM}   \\ 
  \mathcal{G}_{{H}_m}&=\bigg\{
\{\mathcal{F}_{11},
\hdots,\mathcal{F}_{1\bar{p}},\hdots,\mathcal{F}_{r1},\hdots,\mathcal{F}_{r\bar{p}}\} : \mathcal{F}_{ij} \in [0,g_{b_{ij}}]
  \bigg\} \label{thm 2 eq GRM} .
\end{align}
The remaining terms remain the same as the one outlined in Theorem~\ref{Theorem 1 IFAC LMI}. The gain matrix $L$ is calculated by using $L=P^{-1}R^\top$.
\end{theorem}
\begin{proof}
The subsequent inequalities are achieved through the employment of the new variant of Young's inequality~\eqref{L 3.2} on the term $\mathbb{NL}_1$ and $\mathbb{NL}_2$:
\begin{equation}\label{thm 2 pf e 1}
    \mathbf{NL}_1 \leq 
(\mathbb{U}+\mathbb{Z}\mathbb{H}\Phi)^\top (2\mathbb{Z})^{-1}
(\mathbb{U}+\mathbb{Z}\mathbb{H}\Phi),
\end{equation}
and
\begin{equation}\label{thm 2 pf e 2}
     \mathbf{NL}_2 \leq 
(\mathbb{M}+\mathbb{S}\mathbb{G}\Psi)^\top (2\mathbb{S})^{-1}
(\mathbb{M}+\mathbb{S}\mathbb{G}\Psi),
\end{equation}
where $\mathbb{Z}>0$ and $\mathbb{S}>0$ are defined in~\eqref{eq 11 matrix Z} and~\eqref{eq 11 matrix S}, respectively.
\\From~\eqref{thm 2 pf e 1} and~\eqref{thm 2 pf e 2}, the inequality~\eqref{eq W_cal 5} is true if
\begin{equation}\label{thm 2 pf e 3}
    \Sigma+(\mathbb{U}+\mathbb{Z}\mathbb{H}\Phi)^\top (2\mathbb{Z})^{-1}
(\mathbb{U}+\mathbb{Z}\mathbb{H}\Phi)+(\mathbb{M}+\mathbb{S}\mathbb{G}\Psi)^\top (2\mathbb{S})^{-1}
(\mathbb{M}+\mathbb{S}\mathbb{G}\Psi)\leq 0.
\end{equation}

The inequalities~\eqref{eq 7 f_tilde} and~\eqref{eq 7 g_tilde} infer that each element $f_{ij}$ and $g_{ij}$ inside $\mathbb{V}$ and $\mathbb{N}$, respectively, are bounded and belongs to its respective convex sets, whose vertices are described in~\eqref{thm 2 eq FHM} and~\eqref{thm 2 eq GRM}, respectively. Thus, the condition specified in~\eqref{thm 2 pf e 3} is fulfilled if
\begin{equation}\label{thm 2 pf eq 4}
\begin{split}
\Sigma&+\Bigg[(\mathbb{U}+\mathbb{Z}\mathbb{H}\Phi)^\top (2\mathbb{Z})^{-1}
(\mathbb{U}+\mathbb{Z}\mathbb{H}\Phi)\Bigg]_{\Phi \in \mathcal{F}_{H_m}}\\&+\Bigg[(\mathbb{M}+\mathbb{S}\mathbb{G}\Psi)^\top (2\mathbb{S})^{-1}
(\mathbb{M}+\mathbb{S}\mathbb{G}\Psi)\Bigg]_{\Psi \in \mathcal{G}_{H_m}}\leq 0.
\end{split}
\end{equation}
The Schur's compliment of~\eqref{thm 2 pf eq 4} resulted in the LMI~\eqref{eq LMI 2}. From convexity principle proposed in~\cite{boyd1994linear}, the error dynamic~\eqref{eq 8} holds $\mathcal{H}_\infty$ criterion~\eqref{eq H criterion 1} if the LMI~\eqref{eq LMI 2} is solved for all
$\Phi\in\mathcal{F}_{{H}_m}$ and  $\Psi\in\mathcal{G}_{{H}_m}$.
\end{proof}
\section{Comments related to the proposed techniques}\label{sec 5 comment}
In this section, we have outlined some remarks related to the established methodology.
\subsection{Case of the nonlinear systems with linear outputs}
This segment focuses on the observer design for the nonlinear systems having linear outputs in the presence of disturbances/noise. The system~\eqref{eq 1} with linear output is reformulated as:
\begin{align}\label{sec 5 eq 1}
\begin{split}
     x_{k+1}  &= A x_k + G f(x_k)+B u_k+ E \omega_k, \\
     y_k       &= C x_k+D \omega_k, 
\end{split}
\end{align}
where all variables and parameters remain consistent with those specified in~\eqref{eq 1}. The nonlinear function $f(\cdot)$ is presumed to be globally Lipschitz and holds the detailed form~\eqref{eq 2 f}. Analogous to the previous Section~\ref{sec 3 prob statement}, the states of the system~\eqref{sec 5 eq 1} are estimated by deploying the ensuing observer:
\begin{equation}\label{sec5 eq 3}
\hat{x}_{k+1}=A\hat{x}_k+G f(\hat{x}_k)+B u_k+L\big(y_k-C\hat{x}_k\big),
\end{equation}
where all the parameters and variables are the same as the one illustrated in~\eqref{eq 3}. If one follows the steps~\eqref{eq 4}-~\eqref{eq 8} showcased in Section~\ref{sec 3 prob statement}, it is easy to obtain the subsequent error dynamics of the observer~\eqref{sec5 eq 3}:
\begin{equation}\label{sec5 eq 8}
\begin{split}
  e_{k+1}=(A-LC)e_k+ \sum_{i,j=1}^{m,\bar{n}}   f_{ij}G\mathcal{H}_{ij}F_i e_k+(E-LD)\omega_k.  
\end{split}
 \end{equation}
The following corollaries present two new LMI conditions which guarantee the $\mathcal{H}_\infty$ stability of the closed-loop system~\eqref{sec5 eq 8}.
\begin{corollary}
If there exist matrices $P > 0 \in \mathbf{S}^{n}$, $R \in \mathbf{S}^{p \times n}$,
along with $\mathbb{Z}$ in the form of~\eqref{eq 11 matrix Z} and a positive scalar $\mu$, such that, the ensuing optimization problem is solvable:
\begin{equation}\label{eq LMI 1.1}
    \begin{split}
\text{minimize}~\mu~\text{subject~to}&\\
        \begin{bmatrix}
           \Sigma & \mathbb{U}^\top &
    (\mathbb{Z}\mathbb{H}\Phi_m)^\top \\
           \star&-\mathbb{Z}&\mathbb{O}\\
\star&\star&-\mathbb{Z}
        \end{bmatrix}\leq 0,
    \end{split}
\end{equation}
where all variables and parameters are the same as the one described in the LMI~\eqref{eq LMI 1}. Then, the estimation error dynamic~\eqref{sec5 eq 8} satisfied $\mathcal{H}_\infty$ criterion~\eqref{eq H criterion 1}.
\end{corollary}
\begin{corollary}
Let us introduce the matrices $P > 0 \in \mathbf{S}^{n}$, $R \in \mathbf{S}^{p \times n}$, the matrix $\mathbb{Z}$ defined by~\eqref{eq 11 matrix Z}, a positive scalar $\mu$ and the following optimization problem:
\begin{equation}\label{eq LMI 2.1}
\begin{split}
\text{minimize}~\mu~\text{subject~to}&\\
    \begin{bmatrix}
    \Sigma & (\mathbb{U}+\mathbb{Z}\mathbb{H}\Phi)^\top\\
        \star &-2\mathbb{Z}
\end{bmatrix} &< 0,\,\,{ \forall \Phi\in\mathcal{F}_{{H}_m}},
\end{split}
\end{equation}
where all the terms and variables remain consistent with the one specified in the LMI~\eqref{eq LMI 2}. If the aforementioned optimization problem is solvable, then the estimation error dynamic~\eqref{sec5 eq 8} is $\mathcal{H}_\infty$ asymptotically stable.
\end{corollary}

For the proof of both corollaries, one can follow the proof of Theorem~\ref{Theorem 1 IFAC LMI} and Theorem~\ref{Theorem 2 LPV LMI}.

\subsection{Case of absence of exogenous input}
At $\omega=0$, the inequality~\eqref{eq H criterion 2} is reformulated as: $$\Delta V_k+ ||e_k||^2 \leq 0.$$ It yields the exponential stability condition $$\Delta V_k \leq -\sigma V(e_k),$$ along with $\sigma=\frac{1}{\lambda_{\max}(P)}>0$. Hence, the proposed LMIs ensure the exponential stability of the error dynamic~\eqref{eq 8} when $\omega=0$.

\section{Illustrative examples}\label{sec 6 num examples}
This section is dedicated to the analysis of the established LMI-based observer methodology.
The first part of this segment emphasises the superiority of the proposed LMI approach through a numerical example. Later on, the performance of the observer is demonstrated by applying it to SoC estimation in Li-ion batteries.
\subsection{Numerical example 1}
Let us consider a nonlinear system represented in the form~\eqref{sec 5 eq 1} with the ensuing parameters:
$A=\begin{bmatrix}
    0 & 1& 0\\0&  -1& 1\\0 & -1 &1
\end{bmatrix}$, $G=\begin{bmatrix}
    1 & 0\\0&0\\  0& 1
\end{bmatrix}$, $C=\begin{bmatrix}
    1 &0&0\\ 0&0&1
\end{bmatrix}$, $B=E=\begin{bmatrix}
    1 \\1\\ 1
\end{bmatrix}$ and $D=\begin{bmatrix}
    1 \\-0.1
\end{bmatrix}$. In addition to this, $f(x)=\begin{bmatrix}
    \sin (\theta_1 x_1)\\\cos (\theta_2 x_2)
\end{bmatrix}$ along with $H_1=\begin{bmatrix}
    1& 0 &-1\\-1& 0 &1\\1& 0& 0
\end{bmatrix}$ and $H_2=\begin{bmatrix}
    0 &-1& 1\\1& 1& 0\\-1& 0& 0
\end{bmatrix}$. Hence, we get $m=2$ and $\bar{n}=3$.

Further, the developed LMI conditions~\eqref{eq LMI 1.1} and~\eqref{eq LMI 2.1} are solved using MATLAB toolbox to compute the observer parameter $L$ and the optimal noise attenuation index $\sqrt{\mu}$ for different values of $\theta_1$ and $\theta_2$. The obtained minimised value $\sqrt{\mu}$ is outlined in Table~\ref{tab ex 1}. 
The proposed LMI provides a better noise attenuation level than the one obtained from~\cite[LMI~(50)]{Khadilja_2019_rnc} and~\cite[LMI~(45)]{gasmi_2016} for all considered values of $\theta_1$ and $\theta_2$. It implies that the noise mitigation achieved by the developed LMI-based observer is more efficient than the existing methods. Hence, it emphasizes that the established LMI condition gives better results than the LMIs of~\cite{Khadilja_2019_rnc} and~\cite{gasmi_2016}. 
Therefore, the effectiveness of the derived LMI condition is highlighted through this comparison.
\begin{table}[!ht]
\centering
\caption{Portraying the optimal values of $\sqrt{\mu}$ in several cases}
\label{tab ex 1}
\renewcommand{\arraystretch}{1.1}
\begin{tabular}{|c|c|c|c|c|c|}
\hline
\begin{tabular}[c]{@{}c@{}}LMI\\ approaches\end{tabular}  &\begin{tabular}[c]{@{}c@{}}$\theta_1=0.1$\\ $\theta_2=0.5$\end{tabular}&\begin{tabular}[c]{@{}c@{}}$\theta_1=0.2$\\ $\theta_2=0.1$\end{tabular}&\begin{tabular}[c]{@{}c@{}}$\theta_1=0.2$\\ $\theta_2=0.4$\end{tabular}&\begin{tabular}[c]{@{}c@{}}$\theta_1=0.4$\\ $\theta_2=0.1$\end{tabular}&\begin{tabular}[c]{@{}c@{}}$\theta_1=0.5$\\ $\theta_2=0.5$\end{tabular}\\
\hline
    LMI~\eqref{eq LMI 1.1}&
    \begin{tabular}[c]{@{}c@{}}\textcolor{plum}{$3.4167$}\\ \textcolor{plum}{$\times 10^{-6}$}\end{tabular}&
    \begin{tabular}[c]{@{}c@{}}\textcolor{plum}{$1.9600$}\\ \textcolor{plum}{$\times 10^{-6}$}\end{tabular}& 
    \begin{tabular}[c]{@{}c@{}}\textcolor{plum}{$1.1503$}\\ \textcolor{plum}{$\times 10^{-6}$}\end{tabular}
    &
    {$1.7356$}&{inf} \\ \hline
    LMI~\eqref{eq LMI 2.1}&
     \begin{tabular}[c]{@{}c@{}}\textcolor{mypink}{$1.2145$}\\ \textcolor{mypink}{$\times 10^{-6}$}\end{tabular}
  &\begin{tabular}[c]{@{}c@{}}\textcolor{mypink}{$1.0941$}\\ \textcolor{mypink}{$\times 10^{-6}$}\end{tabular} & 
  \begin{tabular}[c]{@{}c@{}}\textcolor{mypink}{$3.2263$}\\ \textcolor{mypink}{$\times 10^{-8}$}\end{tabular} &
   \begin{tabular}[c]{@{}c@{}}\textcolor{mypink}{$1.9462$}\\ \textcolor{mypink}{$\times 10^{-6}$}\end{tabular} 
&\begin{tabular}[c]{@{}c@{}}\textcolor{mypink}{$2.1295$}\\ \textcolor{mypink}{$\times 10^{-6}$}\end{tabular}\\ \hline
    ~\cite[LMI~(50)]{Khadilja_2019_rnc}  & 0.0393  &0.0994& 0.1438&\textcolor{plum}{0.4885}&\textcolor{plum}{2.4415}
     \\ \hline
      ~\cite[LMI~(45)]{gasmi_2016}  & 2.1801  &2.1149 & 2.3706& 2.4415 &4.4248\\\hline
\end{tabular}
\end{table}
In the sequel, the validation of observer performance is showcased.

\subsection{Application: SoC estimation of Li-ion batteries} 

In this segment, the authors have demonstrated the effectiveness of the established observer methodology through the deployment of the developed observer for the State-of-Charge (SoC) estimation of the Li-ion battery model. Let us consider the subsequent $2^{\text{nd}}$ order equivalent circuit model proposed in~\cite{2017_SoC_model_paramters}:
\begin{equation}\label{sec 4 battery model}
    \begin{split}
         \dot{V}_1&=\frac{1}{R_1 C_1} V_1+\frac{1}{C_1} I,\\
        \dot{V}_2&=\frac{1}{R_2 C_2} V_2+\frac{1}{C_2} I,\\
        \dot{s}&=-\frac{1}{C_n} s,\\
        V_t&=OCV(s)-V_1-V_2+R_0 I
    \end{split}
\end{equation}
where $V_1$ and $V_2$ represent voltages across polarisation resistances $R_1$ and $R_2$, respectively, whereas $s$ indicates the SoC of the battery. $I$ infers the current
flowing through the load. The terms $C_1$ and $C_2$ denote polarisation capacitors, however, $C_n$ depict the capacity of the battery. The terminal voltage ($V_t$) is considered as the output of the model. The term $OCV(s)$ depicts the open circuit voltage (OCV) of the battery, and it is illustrated as: 
\begin{equation}\label{SOC OCV relation}
    OCV(s)=0.9206\cdot s^3  - 1.3781\cdot s^2+1.3905\cdot s+3.2416.
\end{equation}
The details of the remaining parameters of the model are as follows:
\begin{itemize}
    \item Battery capacity: $C_n=5~\si{A\cdot hour}$; Battery resistance: $R_0=0.0314~\si{\Omega}$; Sampling time: $T_s=0.02~\si{hour}$;
    \item Polarisation resistances: $R_1=0.0181~\si{\Omega},\,\,R_2=0.0281~\si{\Omega}$;
    \item Polarisation capacitors: $C_1=1712~\si{F},\,\,C_2=55257~\si{F}$;
\end{itemize}
Through the utilization of the Euler forward method, the system~\eqref{sec 4 battery model} is transformed into the discrete-time model represented~\eqref{eq 1}, whose parameters are showcased as: 
$x(k)=\begin{bmatrix}
    V_1(k)\\V_2(k)\\ s(k)
\end{bmatrix},\, A=\begin{bmatrix}
    1-\frac{T_s}{R_1 C_1} & 0&0\\
    0&1-\frac{T_s}{R_2 C_2}&0\\
    0&0&1
\end{bmatrix},\,\,B_1=\begin{bmatrix}
    \frac{T_s}{C_1} \\
    \frac{T_s}{C_2}\\
    -\frac{T_s}{C_n}
\end{bmatrix},\,\,B_2=-R_0,\,\,F=1,\,\,g(x(k))=OCV (x_3(k))$, $E=B_1$ and $D=1$. Let us assume that the system dynamics and outputs are corrupted with the Gaussian noise $\omega \leadsto (0,0.1)$.
\\Since $0\leq x_3(k) \leq 1$, it is easy to infer that the partial derivative of $g(x)$ satisfies~\eqref{eq 7 g_tilde}. 

The developed LMI~\eqref{eq LMI 2} is solved by using YALMIP toolbox~\cite{YALMIP_ref_1393890}, and we obtain:
$$\sqrt{\mu}=4.0907\times 10^{-6}~\text{and}~L=\begin{bmatrix}
    -0.0005\\
   -0.0000\\
    0.1769
\end{bmatrix}.$$ 

\begin{figure}
    \centering
\includegraphics[width=\linewidth]{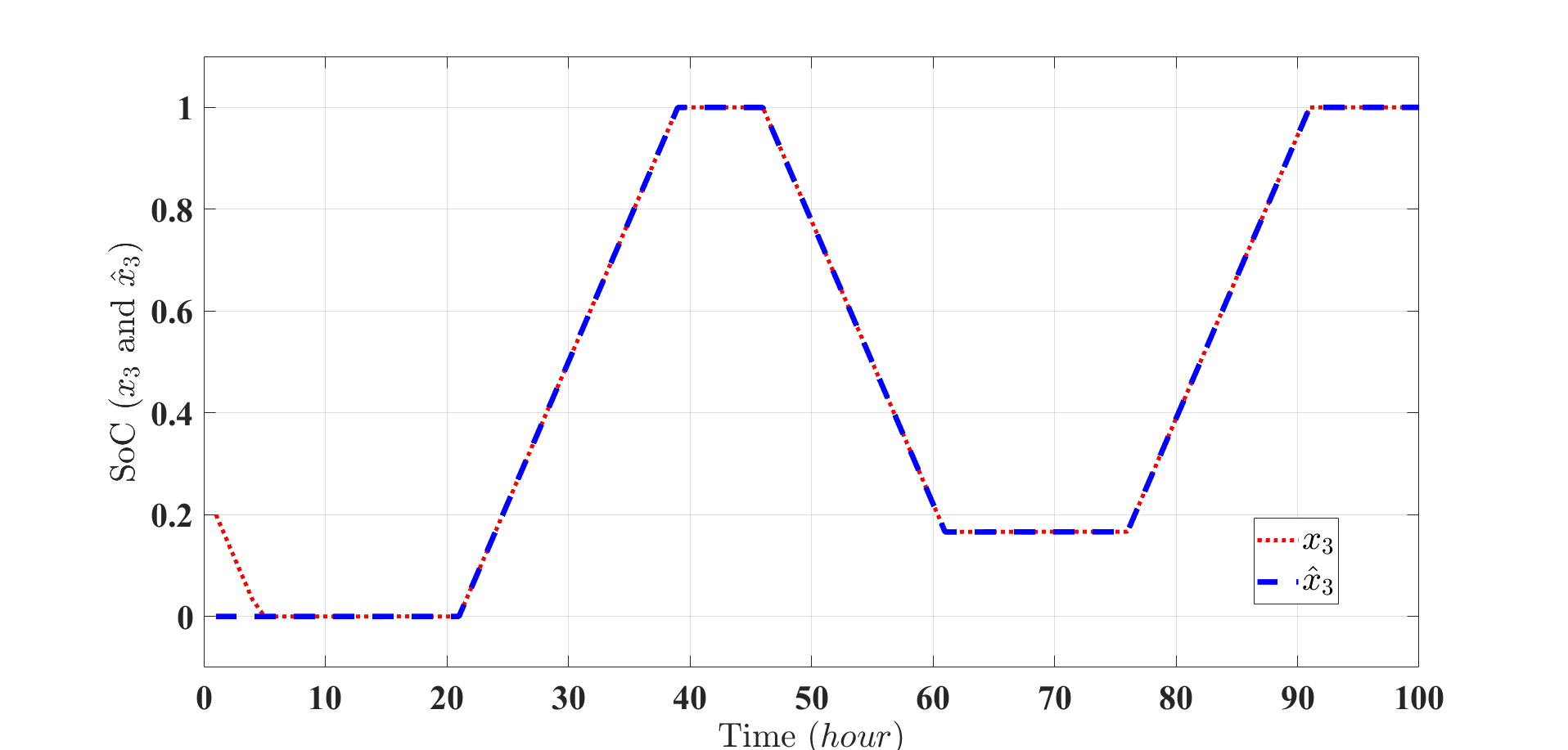}
    \caption{Graph of estimated and actual SoC}
    \label{fig 1 SOC}
\end{figure}

By using the aforementioned matrix $L$, the proposed observer~\eqref{eq 3} is implemented in MATLAB environment for the state estimation purpose. The plot of estimated SoC and actual SoC is shown in Figure~\ref{fig 1 SOC}. It shows that the developed observer performs the efficient estimation of SoC along with optimal noise attenuation. Further, for the same task, the extended Kalman filter (EKF) methodology of~\cite{EKF_result} is implemented in MATLAB. The RMSE values of the estimation error obtained from EKF and the established observer are summarised in Table~\ref{tab ex 2}. It emphasises that the accuracy of the SoC obtained from the proposed approach is relatively better than the one achieved from EKF. Thus, the effectiveness of the new LMI-based observer over the well-known EKF technique is highlighted.
\begin{table}[!ht]
\centering
\caption{Comparison of RMSE values of the estimation error}
\label{tab ex 2}
\renewcommand{\arraystretch}{1.3}
\begin{tabular}{|c|c|c|c|}
\hline
Methodology  &$\tilde{x}_1$&$\tilde{x}_2$&$\tilde{x}_3$ (\textcolor{mygreen}{SoC})\\ \hline
   Proposed observer~\eqref{eq 3} with LMI~\eqref{eq LMI 2}&\textcolor{myred1}{$7.71\times 10^{-4}$}  &\textcolor{myred1}{$1.97\times 10^{-7}$}  & \textcolor{myred1}{0.0014} \\ \hline
    EKF approach~\cite{EKF_result}  & {$7.78\times 10^{-4}$}  &{$7.92\times 10^{-4}$}  & {0.0028}
     \\ \hline
\end{tabular}
\end{table}   
\section{Conclusion}\label{sec 7 conclusion}

In this paper, the problem of nonlinear observer design for discrete-time systems is addressed. It is tackled by formulating two new LMI conditions which provide the observer gain and ensure the $\mathcal{H}_\infty$ stability of the estimation error of the proposed observer. The established matrix-multiplier-based LMIs are derived by incorporating the reformulated Lipschitz property, a new variant of Young inequality. 
The obtained LMIs encompass some additional decision variables as compared to the existing ones due to the deliberate use of matrix multipliers and a new variant of Young inequality. It resulted in an improvement in LMI feasibility.
Thus, the introduction of generalized matrix multipliers inside LMIs plays a vital role in their enhancement. Further, the performance of the observer and the efficacy of the LMI are demonstrated by using numerical examples. From a future perspective, the authors plan to implement the proposed strategy for the stabilization of the same class of systems used in this article.

\bibliographystyle{elsarticle-num-names}
 \bibliography{cite}

\begin{thebibliography}{22}
\expandafter\ifx\csname natexlab\endcsname\relax\def\natexlab#1{#1}\fi
\providecommand{\url}[1]{\texttt{#1}}
\providecommand{\href}[2]{#2}
\providecommand{\path}[1]{#1}
\providecommand{\DOIprefix}{doi:}
\providecommand{\ArXivprefix}{arXiv:}
\providecommand{\URLprefix}{URL: }
\providecommand{\Pubmedprefix}{pmid:}
\providecommand{\doi}[1]{\href{http://dx.doi.org/#1}{\path{#1}}}
\providecommand{\Pubmed}[1]{\href{pmid:#1}{\path{#1}}}
\providecommand{\bibinfo}[2]{#2}
\ifx\xfnm\relax \def\xfnm[#1]{\unskip,\space#1}\fi
\bibitem[{Shen(2023)}]{shen_2023_SoC_model}
\bibinfo{author}{Y.~Shen},
\newblock \bibinfo{title}{A robust method for state of charge estimation of lithium-ion batteries using adaptive nonlinear neural observer},
\newblock \bibinfo{journal}{Journal of Energy Storage} \bibinfo{volume}{72} (\bibinfo{year}{2023}) \bibinfo{pages}{108480}.
\bibitem[{Zemouche et~al.(2017)Zemouche, Rajamani, Phanomchoeng, Boulkroune, Rafaralahy, and Zasadzinski}]{zemouche2017circle}
\bibinfo{author}{A.~Zemouche}, \bibinfo{author}{R.~Rajamani}, \bibinfo{author}{G.~Phanomchoeng}, \bibinfo{author}{B.~Boulkroune}, \bibinfo{author}{H.~Rafaralahy}, \bibinfo{author}{M.~Zasadzinski},
\newblock \bibinfo{title}{Circle criterion-based $\mathcal{H}_\infty$ observer design for {L}ipschitz and monotonic nonlinear systems-enhanced {LMI} conditions and constructive discussions},
\newblock \bibinfo{journal}{Automatica} \bibinfo{volume}{85} (\bibinfo{year}{2017}) \bibinfo{pages}{412--425}.
\bibitem[{Gauthier et~al.(1992)Gauthier, Hammouri, and Othman}]{transformation-based}
\bibinfo{author}{J.~Gauthier}, \bibinfo{author}{H.~Hammouri}, \bibinfo{author}{S.~Othman},
\newblock \bibinfo{title}{A simple observer for nonlinear systems applications to bioreactors},
\newblock \bibinfo{journal}{IEEE Transactions on Automatic Control} \bibinfo{volume}{37} (\bibinfo{year}{1992}) \bibinfo{pages}{875--880}. \DOIprefix\doi{10.1109/9.256352}.
\bibitem[{Ahrens and Khalil(2009)}]{ahrens2009high}
\bibinfo{author}{J.~H. Ahrens}, \bibinfo{author}{H.~K. Khalil},
\newblock \bibinfo{title}{High-gain observers in the presence of measurement noise: A switched-gain approach},
\newblock \bibinfo{journal}{Automatica} \bibinfo{volume}{45} (\bibinfo{year}{2009}) \bibinfo{pages}{936--943}.
\bibitem[{Spurgeon(2008)}]{spurgeon2008sliding}
\bibinfo{author}{S.~K. Spurgeon},
\newblock \bibinfo{title}{Sliding mode observers: a survey},
\newblock \bibinfo{journal}{International Journal of Systems Science} \bibinfo{volume}{39} (\bibinfo{year}{2008}) \bibinfo{pages}{751--764}.
\bibitem[{Zhang et~al.(2012)Zhang, Su, Zhu, and Yue}]{2012_Nonlinear_Discrete_LMI}
\bibinfo{author}{W.~Zhang}, \bibinfo{author}{H.~Su}, \bibinfo{author}{F.~Zhu}, \bibinfo{author}{D.~Yue},
\newblock \bibinfo{title}{A note on observers for discrete-time lipschitz nonlinear systems},
\newblock \bibinfo{journal}{IEEE Transactions on Circuits and Systems II: Express Briefs} \bibinfo{volume}{59} (\bibinfo{year}{2012}) \bibinfo{pages}{123--127}. \DOIprefix\doi{10.1109/TCSII.2011.2174671}.
\bibitem[{Hu et~al.(2012)Hu, Wang, Niu, and Stergioulas}]{hu2012h}
\bibinfo{author}{J.~Hu}, \bibinfo{author}{Z.~Wang}, \bibinfo{author}{Y.~Niu}, \bibinfo{author}{L.~K. Stergioulas},
\newblock \bibinfo{title}{$\mathcal{H}_\infty$ sliding mode observer design for a class of nonlinear discrete time-delay systems: A delay-fractioning approach},
\newblock \bibinfo{journal}{International Journal of Robust and Nonlinear Control} \bibinfo{volume}{22} (\bibinfo{year}{2012}) \bibinfo{pages}{1806--1826}.
\bibitem[{Ibrir(2007)}]{ibrir2007circle_descrete}
\bibinfo{author}{S.~Ibrir},
\newblock \bibinfo{title}{Circle-criterion approach to discrete-time nonlinear observer design},
\newblock \bibinfo{journal}{Automatica} \bibinfo{volume}{43} (\bibinfo{year}{2007}) \bibinfo{pages}{1432--1441}.
\bibitem[{Zhang et~al.(2012)Zhang, Su, Zhu, and Yue}]{LMI_discrete_2012}
\bibinfo{author}{W.~Zhang}, \bibinfo{author}{H.~Su}, \bibinfo{author}{F.~Zhu}, \bibinfo{author}{D.~Yue},
\newblock \bibinfo{title}{A note on observers for discrete-time {L}ipschitz nonlinear systems},
\newblock \bibinfo{journal}{IEEE Transactions on Circuits and Systems II: Express Briefs} \bibinfo{volume}{59} (\bibinfo{year}{2012}) \bibinfo{pages}{123--127}. \DOIprefix\doi{10.1109/TCSII.2011.2174671}.
\bibitem[{Zemouche and Boutayeb(2006)}]{zemouche2006observer}
\bibinfo{author}{A.~Zemouche}, \bibinfo{author}{M.~Boutayeb},
\newblock \bibinfo{title}{Observer design for lipschitz nonlinear systems: the discrete-time case},
\newblock \bibinfo{journal}{IEEE Transactions on Circuits and Systems II: Express Briefs} \bibinfo{volume}{53} (\bibinfo{year}{2006}) \bibinfo{pages}{777--781}.
\bibitem[{Abbaszadeh and Marquez(2009)}]{abbaszadeh2009lmi}
\bibinfo{author}{M.~Abbaszadeh}, \bibinfo{author}{H.~J. Marquez},
\newblock \bibinfo{title}{{LMI} optimization approach to robust $\mathcal{H}_\infty$ observer design and static output feedback stabilization for discrete-time nonlinear uncertain systems},
\newblock \bibinfo{journal}{International Journal of Robust and Nonlinear Control: IFAC-Affiliated Journal} \bibinfo{volume}{19} (\bibinfo{year}{2009}) \bibinfo{pages}{313--340}.
\bibitem[{Mohite et~al.(2023)Mohite, Alma, Zemouche, and Haddad}]{Shiv_IFAC_new}
\bibinfo{author}{S.~Mohite}, \bibinfo{author}{M.~Alma}, \bibinfo{author}{A.~Zemouche}, \bibinfo{author}{M.~Haddad},
\newblock \bibinfo{title}{{LMI}-based $\mathcal{H}_\infty$ observer design for nonlinear {L}ipschitz system},
\newblock \bibinfo{journal}{IFAC-PapersOnLine} \bibinfo{volume}{56} (\bibinfo{year}{2023}) \bibinfo{pages}{6745--6750}.
\bibitem[{Zemouche and Boutayeb(2013)}]{zemouche2013lmi}
\bibinfo{author}{A.~Zemouche}, \bibinfo{author}{M.~Boutayeb},
\newblock \bibinfo{title}{On {LMI} conditions to design observers for {L}ipschitz nonlinear systems},
\newblock \bibinfo{journal}{Automatica} \bibinfo{volume}{49} (\bibinfo{year}{2013}) \bibinfo{pages}{585--591}.
\bibitem[{Chaib~Draa et~al.(2019)Chaib~Draa, Zemouche, Alma, Voos, and Darouach}]{Khadilja_2019_rnc}
\bibinfo{author}{K.~Chaib~Draa}, \bibinfo{author}{A.~Zemouche}, \bibinfo{author}{M.~Alma}, \bibinfo{author}{H.~Voos}, \bibinfo{author}{M.~Darouach},
\newblock \bibinfo{title}{A discrete-time nonlinear state observer for the anaerobic digestion process},
\newblock \bibinfo{journal}{International Journal of Robust and Nonlinear Control} \bibinfo{volume}{29} (\bibinfo{year}{2019}) \bibinfo{pages}{1279--1301}.
\bibitem[{Chu and Li(2018)}]{chu2018_SMO}
\bibinfo{author}{X.~Chu}, \bibinfo{author}{M.~Li},
\newblock \bibinfo{title}{$\mathcal{H}_\infty$ observer-based event-triggered sliding mode control for a class of discrete-time nonlinear networked systems with quantizations},
\newblock \bibinfo{journal}{Isa Transactions} \bibinfo{volume}{79} (\bibinfo{year}{2018}) \bibinfo{pages}{13--26}.
\bibitem[{Avil{\'e}s and Moreno(2019)}]{aviles2019observer}
\bibinfo{author}{J.~D. Avil{\'e}s}, \bibinfo{author}{J.~A. Moreno},
\newblock \bibinfo{title}{Observer design for discrete-time nonlinear systems using the stability radii theory},
\newblock \bibinfo{journal}{IEEE Transactions on Circuits and Systems II: Express Briefs} \bibinfo{volume}{67} (\bibinfo{year}{2019}) \bibinfo{pages}{1959--1963}.
\bibitem[{Mohite et~al.(2023)Mohite, Alma, and Zemouche}]{LCSS_Shiv_2023}
\bibinfo{author}{S.~Mohite}, \bibinfo{author}{M.~Alma}, \bibinfo{author}{A.~Zemouche},
\newblock \bibinfo{title}{Observer-based stabilization of {L}ipschitz nonlinear systems by using a new matrix-multiplier-based lmi approach},
\newblock \bibinfo{journal}{IEEE Control Systems Letters} \bibinfo{volume}{7} (\bibinfo{year}{2023}) \bibinfo{pages}{3723--3728}. \DOIprefix\doi{10.1109/LCSYS.2023.3341549}.
\bibitem[{Boyd et~al.(1994)Boyd, El~Ghaoui, Feron, and Balakrishnan}]{boyd1994linear}
\bibinfo{author}{S.~Boyd}, \bibinfo{author}{L.~El~Ghaoui}, \bibinfo{author}{E.~Feron}, \bibinfo{author}{V.~Balakrishnan}, \bibinfo{title}{Linear matrix inequalities in system and control theory}, \bibinfo{publisher}{SIAM}, \bibinfo{year}{1994}.
\bibitem[{Gasmi et~al.(2016)Gasmi, Thabet, Boutayeb, and Aoun}]{gasmi_2016}
\bibinfo{author}{N.~Gasmi}, \bibinfo{author}{A.~Thabet}, \bibinfo{author}{M.~Boutayeb}, \bibinfo{author}{M.~Aoun},
\newblock \bibinfo{title}{Observers for nonlinear {L}ipschitz discrete time systems with extension to $\mathcal{H}_\infty$ filtering design},
\newblock in: \bibinfo{booktitle}{2016 13th International Multi-Conference on Systems, Signals \& Devices (SSD)}, \bibinfo{year}{2016}, pp. \bibinfo{pages}{364--369}. \DOIprefix\doi{10.1109/SSD.2016.7473668}.
\bibitem[{Tian et~al.(2017)Tian, Li, Tian, and Xia}]{2017_SoC_model_paramters}
\bibinfo{author}{Y.~Tian}, \bibinfo{author}{D.~Li}, \bibinfo{author}{J.~Tian}, \bibinfo{author}{B.~Xia},
\newblock \bibinfo{title}{State of charge estimation of lithium-ion batteries using an optimal adaptive gain nonlinear observer},
\newblock \bibinfo{journal}{Electrochimica Acta} \bibinfo{volume}{225} (\bibinfo{year}{2017}) \bibinfo{pages}{225--234}.
\bibitem[{Lofberg(2004)}]{YALMIP_ref_1393890}
\bibinfo{author}{J.~Lofberg},
\newblock \bibinfo{title}{Yalmip : a toolbox for modeling and optimization in matlab},
\newblock in: \bibinfo{booktitle}{2004 IEEE International Conference on Robotics and Automation (IEEE Cat. No.04CH37508)}, \bibinfo{year}{2004}, pp. \bibinfo{pages}{284--289}.
\bibitem[{Chen et~al.(2013)Chen, Fu, and Mi}]{EKF_result}
\bibinfo{author}{Z.~Chen}, \bibinfo{author}{Y.~Fu}, \bibinfo{author}{C.~C. Mi},
\newblock \bibinfo{title}{State of charge estimation of lithium-ion batteries in electric drive vehicles using extended kalman filtering},
\newblock \bibinfo{journal}{IEEE Transactions on Vehicular Technology} \bibinfo{volume}{62} (\bibinfo{year}{2013}) \bibinfo{pages}{1020--1030}. \DOIprefix\doi{10.1109/TVT.2012.2235474}.

\end{thebibliography}
\end{document}